\documentclass[a4paper]{llncs}

\usepackage{xspace}
\usepackage{subcaption}
\usepackage{color}
\usepackage{paralist}
\usepackage{hyperref}
\usepackage{amsmath}
\usepackage{amssymb}
\usepackage{tikz}
\usepackage{stmaryrd}
\usepackage{thm-restate}
\usepackage{mathtools}
\usetikzlibrary{automata,positioning}

\InputIfFileExists{cutmargins.tex}{%
	\pdfpagesattr{/CropBox [125 170 490 728]} 
}{%
}

\newcommand{\yfc}[1]{\textcolor{brown}{\ifmmode \text{[YFC: #1]}\else [YFC: #1] \fi}}

\newcommand{\ol}[1]{\textcolor{blue}{\ifmmode \text{[OL: #1]}\else [OL: #1] \fi}}
\newcommand{\vh}[1]{\textcolor{orange}{\ifmmode \text{[VH: #1]}\else [VH: #1] \fi}}

\newcommand{\claimqed}[0]{\hfill $\blacksquare$}

\newcommand{\langsymb}[0]{\cL}
\newcommand{\langof}[1]{\langsymb(#1)}
\newcommand{\langautof}[2]{\langsymb_{#1}(#2)}
\newcommand{\langaut}[1]{\langautof{\A}{#1}}

\newcommand{\limplies}[0]{\Rightarrow}
\newcommand{\iffdef}[0]{\mathrel{\stackrel{\mathrm{def}}{\iff}}}

\newcommand{\bigO}[0]{\mathcal{O}}
\newcommand{\bigOof}[1]{\bigO(#1)}

\newcommand{\clxof}[1]{[#1]_x} 

\newcommand{\states}{Q}

\newcommand{\trans}{\delta}

\newcommand{\inits}{I}
\newcommand{\acc}{F}

\newcommand{\ltr}[1]{\mathop{\xrightarrow{#1}}}

\newcommand{\pspace}[0]{\textsc{PSpace}}

\newcommand{\costates}[1][\Comp]{\states^O}
\newcommand{\cotrans}[1][\Comp]{\trans^O}
\newcommand{\coinits}[1][\Comp]{\inits^O}
\newcommand{\coacc}[1][\Comp]{\acc^O}

\newcommand{\A}{\mathcal{A}}
\newcommand{\B}{\mathcal{B}}

\newcommand{\cL}{\mathcal{L}}

\newcommand{\cT}{\mathcal{T}}

\newcommand{\rankleq}[0]{\mathrel{\leq^{S}_{a}}}

\newcommand{\subseteqlang}[0]{\mathrel{\subseteq_\langsymb}}
\newcommand{\ftransover}[1]{\overset{#1}{\underset{F}{\leadsto}}}
\newcommand{\transover}[1]{\overset{#1}{\leadsto}}

\newcommand{\cond}[1]{\mathcal{C}^{#1}}
\newcommand{\condof}[2]{\cond{#1}(#2)}
\newcommand{\condtrace}[1]{\condof{#1}{\pi_p, \pi_r}}

\newcommand{\seqof}[2]{#1_{#2}}
\newcommand{\word}[0]{\alpha}
\newcommand{\wordof}[1]{\seqof{\word}{#1}}
\newcommand{\dagg}[0]{\mathcal{G}}
\newcommand{\dagof}[1]{\dagg_{#1}}
\newcommand{\dagw}[0]{\dagof{\word}}
\newcommand{\dagwiof}[1]{\dagw^{#1}}

\newcommand{\simby}[0]{\preceq}
\newcommand{\dirsimby}[0]{\mathrel{\simby_{\mathit{di}}}}
\newcommand{\delsimby}[0]{\mathrel{\simby_{\mathit{de}}}}
\newcommand{\fairsimby}[0]{\mathrel{\simby_{f}}}

\newcommand{\xsimby}[0]{\mathrel{\simby_x}}
\newcommand{\quotxsimby}[0]{\mathrel{\simby^{\similar}_{x}}}

\newcommand{\similar}[0]{\mathrel{\approx}}
\newcommand{\similarof}[1]{\mathrel{\similar_{#1}}}
\newcommand{\xsimilar}[0]{\similarof{x}}

\newcommand{\fairsimilar}[0]{\similarof{f}}

\newcommand{\rank}[0]{\mathit{rank}}
\newcommand{\rankof}[1]{\rank(#1)}
\newcommand{\rankwof}[1]{\rank_{\word}(#1)}

\newcommand{\BSven}[0]{\B_{\mathit{S}}}

\newcommand{\BSvenRankdir}[0]{\BSven^{\mathit{di}}}
\newcommand{\BSvenRankdel}[0]{\BSven^{\mathit{de}}}
\newcommand{\BSvenRankdirdel}[0]{\BSven^{\mathit{di+de}}}
\newcommand{\Comp}[0]{\textsc{Comp}}

\newcommand{\BSvenSat}[0]{\B_{\mathit{Sat}}}
\newcommand{\BSvenSatX}[0]{\B^x_{\mathit{Sat}}}
\newcommand{\BSvenSatde}[0]{\B^{\mathit{de}}_{\mathit{Sat}}}
\newcommand{\BSvenSatdi}[0]{\B^{\mathit{di}}_{\mathit{Sat}}}
\newcommand{\BSvenSatdide}[0]{\B^{\mathit{di+de}}_{\mathit{Sat}}}
\newcommand{\CompSchewe}[0]{\Comp$_{\text{S}}$\xspace}

\newcommand{\PurgeDi}[0]{\textsc{Purge}$_{\mathit{di}}$\xspace}
\newcommand{\PurgeDe}[0]{\textsc{Purge}$_{\mathit{de}}$\xspace}
\newcommand{\PurgeDiDe}[0]{\textsc{Purge}$_{\mathit{di+de}}$\xspace}
\newcommand{\PurgeX}[0]{\textsc{Purge}$_x$\xspace}
\newcommand{\Saturate}[0]{\textsc{Saturate}\xspace}

\newcommand{\Qtwodi}[0]{Q_2^{\mathit{di}}}
\newcommand{\Qtwode}[0]{Q_2^{\mathit{de}}}
\newcommand{\Qtwodide}[0]{Q_2^{\mathit{di+de}}}
\newcommand{\deltasatup}[0]{\delta^{\mathit{Sat}}}
\newcommand{\deltasatdown}[0]{\delta_{\mathit{Sat}}}

\newcommand{\numsetof}[1]{[#1]}

\newcommand{\rond}[1]{\lceil\!\!\lceil #1 \rceil\!\!\rceil}

\newcommand{\torem}[0]{\mathcal{P}}

\newcommand{\toremdir}[0]{\torem_{\mathit{di}}}
\newcommand{\toremdirof}[1]{\toremdir(#1)}
\newcommand{\toremdel}[0]{\torem_{\mathit{de}}}
\newcommand{\toremdelof}[1]{\toremdel(#1)}

\newcommand{\strateg}[0]{\sigma}
\newcommand{\strategx}[0]{\strateg_x}

\newcommand{\strategdel}[0]{\strateg_{\mathit{de}}}

\newcommand{\secref}[1]{Section~\ref{#1}}

\newcommand{\str}[1]{\mathit{cl}[#1]}

\newenvironment{claimproofnoqed}[1]{\par\noindent\underline{Proof:}\space#1}{}

\newcommand{\blinded}[1]{\ifx\blindreview\undefined #1 \else \textcolor{black!65}{[blinded for review]}\fi}

\title{Simulations in Rank-Based\\ B\"{u}chi Automata Complementation\\ (Technical Report)}

\author{\blinded{Yu-Fang Chen\inst{1} \and
  Vojt\v{e}ch Havlena\inst{2} \and
  Ond\v{r}ej Leng\'{a}l\inst{2}}
  }

\institute{\blinded{
  Academia Sinica, Taiwan \and
  FIT,
  IT4I Centre of Excellence,
  Brno University of Technology,
  Czech Republic}
}

\begin{document}

\maketitle

\begin{abstract}
  \vspace{-2mm}
  Complementation of B\"{u}chi automata is an essential technique used in
  some approaches for termination analysis of programs.
  The long search for an optimal complementation construction climaxed with the
  work of Schewe, who proposed a worst-case optimal rank-based procedure that
  generates complements of a size matching the theoretical lower bound
  of~$(0.76n)^n$, modulo a~polynomial factor of~$\bigOof{n^2}$.
  Although worst-case optimal, the procedure in many cases produces automata
  that are unnecessarily large.
  In this paper, we propose several ways of how to use the direct and delayed
  simulation relations to reduce the size of the automaton obtained in the
  rank-based complementation procedure.
  Our techniques are based on either (i)~ignoring macrostates that cannot be
  used for accepting a~word in the complement or (ii)~saturating macrostates
  with simulation-smaller states, in order to decrease their total number.
  We experimentally showed that our techniques can indeed considerably decrease
  the size of the output of the complementation.
\end{abstract}

\vspace{-8.0mm}
\section{Introduction}\label{sec:intro}
\vspace{-0.0mm}

B\"uchi automata (BA) complementation is a fundamental problem in program
analysis and formal verification, from both theoretical and practical angles.
It is, for instance, a~critical step in some approaches for termination
analysis, which is an essential part of establishing total correctness of
programs~\cite{fogarty2009buchi,heizmann2014termination,ChenHLLTTZ18}.
Moreover, BA complementation is used as a~component of decision
procedures of some logics for reasoning about programs, such as S1S
capturing a~decidable fragment of
second-order arithmetic~\cite{buchi1962decision} or the temporal logics ETL and
QPTL~\cite{SistlaVW85}.

%

The study of the BA complementation problem can be traced back to 1962, when
B\"{u}chi introduced his automaton model in the seminal
paper~\cite{buchi1962decision} in the context of a decision procedure for the
S1S fragment of second-order arithmetic.
In the paper, a~doubly exponential complementation algorithm based on the infinite
Ramsey theorem is proposed.
In 1988, Safra~\cite{safra1988complexity} introduced a~complementation procedure
with an~$n^{\bigOof{n}}$ upper bound and, in the same year,
Michel~\cite{michel1988complementation} established an~$n!$ lower
bound.
From the traditional theoretical point of view, the problem was already solved,
since exponents in the two bounds matched under the $\bigO$~notation (recall that
$n!$ is approximately $(n / e)^n$).
From a~more practical point of view, a~linear factor in an~exponent has
a~significant impact on real-world applications.
It was established that the upper bound of Safra's construction is $2^{2n}$, so
the hunt for an optimal algorithm continued~\cite{vardi2007buchi}.
A~series of research efforts participated in narrowing the
gap~\cite{KupfermanV01,FriedgutKV06,vardi2008automata,kahler2008complementation,yan2006lower}.
The long journey climaxed with the result of Schewe~\cite{Schewe09}, who
proposed an optimal rank-based procedure that generates complements
of a~size matching the theoretical lower bound of~$(0.76n)^n$ found by
Yan~\cite{yan2006lower}, modulo a~polynomial factor of~$\bigOof{n^2}$.

Although the algorithm of Schewe is worst-case optimal, it often generates
unnecessarily large complements. The standard approach to alleviate this problem
is to decrease the size of the input BA before the complementation starts. Since
minimization of (nondeterministic) BAs is a~$\pspace$-complete problem, more
lightweight reduction methods are necessary. The most prevalent approaches are
those based on various notions of \emph{simulation-based reduction}, such as
reductions based on \emph{direct
simulation}~\cite{bustan2003simulation,somenzi2000efficient}, a~richer
\emph{delayed simulation}~\cite{EtessamiWS05}, or their \emph{multi-pebble}
variants~\cite{etessami2002hierarchy}. These approaches first compute
a~simulation relation over the input BA---which can be done with the time
complexity
$\bigOof{mn}$~\cite{HenzingerHK95,IlieNY04,RanzatoT07,RanzatoT10,Cece17} and
$\bigOof{mn^3}$~\cite{EtessamiWS05} for direct and delayed simulation
respectively, with the number of states~$n$ and transitions~$m$---and then
construct a~\emph{quotient} BA by merging simulation-equivalent states, while
preserving the language of the input BA.
The other approach is a reduction
based on \emph{fair simulation}~\cite{gurumurthy-fair-mini}. The fair simulation
cannot, however, be used for quotienting, but still it can be used for merging
certain states and removing transitions.
The reduced BA is used as the input of
the complementation, which often significantly reduces the size of the result.

In this paper, we propose several ways of how to exploit the direct and delayed
simulations in BA complementation even further to obtain smaller complements
and shorter running times.
We focus, in particular, on the optimal \emph{rank-based} complementation
procedure of Schewe~\cite{Schewe09}.
Essentially, the rank-based construction is an extension of traditional subset
construction for determinizing finite automata, with some additional information
kept in each macrostate (a state in the complemented BA) to track the
acceptance condition of all runs of the input automaton on a~given word.
In particular, it stores the \emph{rank} of each state in a~macrostate, which,
informally, measures the distance to the last accepting state on the
corresponding run in the input~BA.
The main contributions of this paper are the following optimisations of
rank-based complementation for BAs, for an input BA~$\A$ and the
output of the rank-based complementation algorithm~$\B$.
\begin{enumerate}
  \item  \emph{Purging}: We use simulation relations over~$\A$ to remove some
    useless macro\-states during the construction of~$\B$.
    In particular, if a~state~$p$ is simulated by~$q$ in~$\A$, this puts
    a~restriction on the relation between the ranks of runs from~$p$ and
    from~$q$.
    As a~consequence, macrostates that assign ranks violating this
    restriction can be purged from~$\B$.

  \item  \emph{Saturation}: We saturate macrostates with states that are
    simulated by the macrostate; this can reduce the total number of states
    of~$\B$ because two or more macrostates can be mapped to a~single saturated
    macrostate.
    This is inspired by the technique of Glabbeek and Ploeger that uses
    \emph{closures} in finite automata determinization~\cite{GlabbeekP08}.
\end{enumerate}

\enlargethispage{3mm}

\noindent
The proposed optimizations are orthogonal to simulation-based size reduction
mentioned above.
Since the quotienting methods are based on taking only the symmetric fragment of
the simulation, i.e., they merge states that simulate \emph{each other}, after
the quotienting, there might still be many pairs where the simulation holds in
only one way, and can therefore be exploited by our techniques.
Since the considered notions of simulation-based quotienting preserve the
respective simulations, our techniques can be used to optimize the
complementation \emph{at no additional cost}.
Our experimental evaluation of the optimizations showed that in many cases, they
indeed significantly reduce the size of the complemented BA.

\vspace{-0.0mm}
\section{Preliminaries}\label{sec:prelims}
\vspace{-2.0mm}


We fix a~finite nonempty alphabet~$\Sigma$ and the first infinite ordinal
$\omega = \{0, 1, \ldots\}$.
For $n\in\omega$, by $\numsetof{n}$ we denote the
set $\{ 0, \dots, n \}$.
An (infinite) word~$\word$ is represented as a~function
$\word:\omega \to \Sigma$ where the $i$-th symbol is denoted as $\wordof i$.
A~finite word~$w$ of length $n+1$ is represented as a~function $w: \numsetof{n}
\to \Sigma$.
The finite word of length~$0$ is denoted as~$\epsilon$.
We~abuse
notation and sometimes also represent~$\word$ as an~infinite sequence $\word =
\wordof 0 \wordof 1 \dots$ and $w$ as a finite sequence $w = w_0\dots w_{n-1}$.
The suffix $\wordof i \wordof{i+1} \ldots$ of~$\word$ is denoted
by~$\wordof {i:\omega}$.
We~use~$\Sigma^\omega$ to denote the set of all infinite words over~$\Sigma$
and~$\Sigma^*$ to denote the set of all finite words.
For $L\subseteq \Sigma^*$ we define $L^* = \{
u\in\Sigma^*~|~u=w_1\cdots w_n \wedge \forall 1 \leq i \leq n: w_i\in L \}$ and
$L^\omega = \{ \word\in\Sigma^\omega~|~\alpha=w_1w_2\cdots \wedge \forall i \geq
1: w_i\in L \}$ (note that $\{\epsilon\}^\omega = \emptyset$).
Given $L_1, L_2 \subseteq \Sigma^*$, we use $L_1 L_2$ to denote the set
$\{w_1 w_2 \mid w_1 \in L_1, w_2 \in L_2\}$.

A~(nondeterministic) \emph{B\"{u}chi automaton} (BA) over~$\Sigma$
is a~quadruple $\A = (Q, \delta, I, F)$ where $Q$ is a~finite set of
\emph{states}, $\delta$ is a~\emph{transition function} $\delta : Q \times
\Sigma \to 2^Q$, and $I, F \subseteq Q$ are the sets of \emph{initial} and
\emph{accepting} states respectively.
We sometimes treat~$\delta$ as a~set of transitions $p \ltr a q$, for instance,
we use $p \ltr a q \in \delta$ to denote that $q \in \delta(p, a)$.
Moreover, we extend $\delta$ to sets of states $P
\subseteq Q$ as $\delta(P, a) = \bigcup_{p \in P} \delta(p,a)$. A~\emph{run}
of~$\A$ from~$q \in Q$ on an input word $\alpha$ is an infinite sequence $\rho:
\omega \to Q$ that starts in~$q$ and respects~$\delta$, i.e., $\rho_0 = q$ and
$\forall i \geq 0: \rho_i \ltr{\wordof i}\rho_{i+1} \in \delta$.
We say that~$\rho$
is accepting iff it contains infinitely many occurrences of some accepting
state, i.e., $\exists q_f \in F: |\{i \in \omega \mid \rho_i = q_f\}| = \omega$.
A~word $\word$ is accepted by~$\A$ from a~state~$q \in Q$ if there is an
accepting run $\rho$ of $\A$ from~$q$, i.e., $\rho_0 = q$. The set
$\langautof{\A} q = \{\word \in \Sigma^\omega \mid \A \text{ accepts } \word
\text{ from } q\}$ is called the \emph{language} of~$q$ (in~$\A$). Given a~set
of states~$R \subseteq Q$, we define the language of~$R$ as $\langautof \A R =
\bigcup_{q \in R} \langautof \A q$ and the language of~$\A$ as~$\langof \A =
\langautof \A I$. For a~pair of states~$p$ and~$q$ in~$\A$, we use $p
\subseteqlang q$ to denote $\langaut p \subseteq \langaut q$.

Without loss of generality, in this paper, we assume~$\A$ to be complete, i.e., for every state~$q$ and
symbol~$a$, it holds that $\delta(q, a) \neq \emptyset$.
A~\emph{trace} over a~word~$\word$ is an infinite sequence $\pi = q_0
\ltr{\wordof 0} q_1 \ltr{\wordof 1} \cdots$ such that $\rho = q_0 q_1 \ldots$ is
a~run of~$\A$ over~$\word$ from~$q_0$.
We say~$\pi$ is \emph{fair} if it contains infinitely many accepting states.
Moreover, we use $p \transover{w} q$ for $w \in \Sigma^*$ to denote that~$q$ is
reachable from~$p$ over the word~$w$; if a~path from~$p$ to~$q$ over~$w$ contains an
accepting state, we can write $p \ftransover{w} q$.
In this paper, we fix a~complete BA $\A = (Q, \delta, I, F)$.

\vspace{-3.0mm}
\subsection{Simulations}\label{sec:sim}
\vspace{-1.0mm}

We introduce simulation relations between states of a~BA $\A$ using the game
semantics in a~similar manner as in the extensive study of Clemente and
Mayr~\cite{MayrC13}.
In particular, in a~\emph{simulation game} between two players (called Spoiler
and Duplicator) in~$\A$ from a pair of states $(p_0, r_0)$, for any (infinite) trace
over a~word~$\word$ that Spoiler takes starting from~$p_0$, Duplicator tries to
mimic the trace starting from~$r_0$.
On the other hand, Spoiler tries to find a~trace that Duplicator cannot mimic.
The game starts in the configuration $(p_0, r_0)$ and every $i$-th round
proceeds by, first, Spoiler choosing a~transition $p_i \ltr{\wordof i} p_{i+1}$ and,
second, Duplicator mimicking Spoiler by choosing a~matching transition $r_i
\ltr{\wordof i} r_{i+1}$ over the same symbol~$\wordof i$.
The next game configuration is $(p_{i+1}, r_{i+1})$.
Suppose that $\pi_p = p_0 \ltr{\wordof 0} p_1 \ltr{\wordof 1} \cdots$ and $\pi_r =
r_0 \ltr{\wordof 0} r_1 \ltr{\wordof 1} \cdots$ are the two (infinite) traces
constructed during the game.
Duplicator \emph{wins} the simulation game if $\condtrace{x}$ holds,
where $\condtrace{x}$ is a condition that depends on the particular
simulation.
In the current paper, we consider the following simulation relations:
%
%
%
%
%
\vspace{-1.5mm}
\begin{itemize}
  \item  \textbf{direct}~\cite{DillHW92}:
    \hspace*{6.3mm}
    $\condtrace{\mathit{di}} \iffdef \forall i: p_i \in \acc \limplies
    r_i \in \acc,$
  \item  \textbf{delayed}~\cite{EtessamiWS05}:
    \hspace*{3mm}
    $\condtrace{\mathit{de}} \iffdef  \forall i: p_i \in \acc \limplies
    \exists k \geq i: r_k \in \acc,$ and
  \item  \textbf{fair}~\cite{HenzingerKR02}:
    \hspace*{11.4mm}
    $\condtrace{f} \iffdef$ if $\pi_p$ is fair, then $\pi_r$ is fair.
\vspace{-1.5mm}
\end{itemize}

A~maximal $x$-simulation relation ${\xsimby} \subseteq \states \times \states$, for $x
\in \{\mathit{di}, \mathit{de}, f \}$, is defined such that
$p \simby_x r$ iff Duplicator has a winning strategy in the simulation game
with the winning condition $\cond{x}$ starting from $(p,r)$.
Formally, we define a~strategy to be a~(total) mapping $\strateg: \states \times
(Q \times \Sigma \times Q) \to Q$ such that
$\strateg(r, p \ltr{a} p') \in \trans(r, a)$, i.e.,
if Duplicator is in state~$r$ and Spoiler selects a transition $p \ltr{a} p'$,
the strategy picks a~state~$r'$ such that $r \ltr a r' \in \delta$ (and
because~$\A$ is complete, such a~transition always exists).
Note that Duplicator cannot look ahead at Spoiler's future moves.
We use $\strategx$ to denote any winning strategy of Duplicator in
the $\cond{x}$~simulation game.
Let~$\strategx$ and~$\strategx'$ be a~pair of winning strategies in the
$\cond x$ simulation game.
We say that $\strategx$ is \emph{dominated} by~$\strategx'$ if for all
states~$p$ and all transitions~$q \ltr a q'$ it holds that $\strategx(p, q \ltr a
q') \xsimby \strategx'(p, q \ltr a q')$,
and that~$\strategx$ is \emph{strictly dominated} by~$\strategx'$ if~$\strategx$
is dominated by~$\strategx'$ and~$\strategx$ does not dominate~$\strategx'$.
A~strategy is~\emph{dominating} if it is not strictly dominated by any other
strategy.
Strategies are also lifted to traces as follows: let $\pi_p$ be as above, then
$\strateg(r_0, \pi_p) = r_0 \ltr{\wordof 0} r_1 \ltr{\wordof 1} \cdots$ where for
all $i \leq 0$ it holds that $\strateg(r_i, p_i \ltr{\wordof i} p_{i+1}) =
r_{i+1}$.
The considered simulation relations form the following hierarchy:
${} \dirsimby {}\subseteq {}\delsimby {}\subseteq {}\fairsimby {}\subseteq
{}\subseteqlang {}$.
Note that every maximal simulation relation is a~preorder, i.e., reflexive and
transitive.





\enlargethispage{3mm}

\vspace{-3.0mm}
\subsection{Run DAGs}\label{sec:rundag}
\vspace{-2.0mm}
In this section, we recall the terminology from~\cite{Schewe09} (which is
a~minor modification of the terminology from~\cite{KupfermanV01}).
We fix the definition of the \emph{run DAG} of~$\A$ over
a~word~$\word$ to be a~DAG (directed acyclic graph) $\dagw = (V,E)$ of
vertices~$V$ and edges~$E$ where
\begin{itemize}
\vspace{-1.5mm}
  \item  $V \subseteq Q \times \omega$ s.t. $(q, i) \in V$ iff there is
    a~run~$\rho$ of $\A$ over $\alpha$ with $\rho_i = q$,
  \item  $E \subseteq V \times V$ s.t.~$((q, i), (q',i')) \in E$ iff $i' = i+1$
    and $q' \in \delta(q, \wordof i)$.
\vspace{-1.5mm}
\end{itemize}

\noindent
Given $\dagw$ as above, we will write $(p, i) \in \dagw$ to denote that $(p, i) \in V$.
We call $(p,i)$ \emph{accepting} if $p$ is an accepting state.
$\dagw$ is \emph{rejecting} if it contains no path with infinitely many
accepting vertices.
A~vertex~$(p, i) \in \dagw$ is \emph{finite} if the set of vertices reachable
from $(p, i)$ is finite, \emph{infinite} if it is not finite, and
\emph{endangered} if $(p,i)$ cannot reach an accepting vertex.

We assign ranks to vertices of run DAGs as follows:
Let $\dagwiof 0 = \dagw$ and~$j = 0$.
Repeat the following steps until the~fixpoint or for at most $2n + 1$ steps,
where $n$~is the number of states of~$\A$.
\begin{itemize}
\vspace{-1.5mm}
  \item  Set $\rankwof{p,i} := j$ for all finite vertices $(p,i)$ of~$\dagwiof j$
    and let $\dagwiof{j+1}$ be $\dagwiof{j}$ minus the vertices with the
    rank~$j$.
  \item  Set $\rankwof{p,i} := j+1$ for all endangered vertices $(p,i)$
    of~$\dagwiof {j+1}$ and let $\dagwiof{j+2}$ be $\dagwiof{j+1}$ minus the
    vertices with the rank~$j+1$.
  \item  Set $j := j + 2$.
\vspace{-1.5mm}
\end{itemize}
%

\noindent
For all vertices~$v$ that have not been assigned a~rank yet, we assign
$\rankwof{v} := \omega$.
(Note that since $\A$ is complete, then $\dagwiof 1 = \dagwiof 0$.)

\vspace{-1mm}
\begin{lemma}\label{lem:rankings}
If $\word \notin \langof \A$, then $0 \leq \rankwof v \leq 2n$ for all $v \in \dagw$.
Moreover, if $\word \in \langof  \A$, then there is a vertex $(p, 0) \in \dagw$
s.t.~$\rankwof{p,0} = \omega$.
\end{lemma}

\begin{proof}
Follows from Corollary~3.3 in \cite{KupfermanV01}.
\qed
\end{proof}

\vspace{-5.0mm}
\section{Complementing B\"{u}chi Automata}\label{sec:label}
\vspace{-2.0mm}

We use as the starting point the complementation procedure of
Schewe~\cite[Section 3.1]{Schewe09}, which we denote as \CompSchewe (the `S'
stands for `Schewe').
The procedure works with the notion of level rankings.
Given $n = |Q|$, a~\emph{(level) ranking} is a~function $f: Q \to \numsetof{2n}$ such that
$\{f(q_f) \mid q_f \in F\} \subseteq \{0, 2, \ldots, 2n\}$, i.e., $f$~assigns even
ranks to accepting states of~$\A$.%
\footnote{%
  Note that our basic definitions slightly differs from the ones in Section~2.3
  of~\cite{Schewe09}.
  This is because of a~typo in~\cite{Schewe09};
  indeed, if the procedure from~\cite{Schewe09} is implemented as
  is, the output does not accept the complement (there might be a macrostate
  $(S, O, f)$ where $S$~contains accepting states and $O$~is empty, and,
  therefore, the whole macrostate is accepting, which is wrong).
}
For a~ranking~$f$, the \emph{rank} of~$f$ is defined as~$\rankof f = \max\{f(q)
\mid q \in Q\}$.
For a~set of states $S \subseteq Q$, we call~$f$ to be $S$-\emph{tight} if
(i)~it has an odd rank~$r$,
(ii)~$\{f(s) \mid s \in S\} \supseteq \{1, 3, \ldots, r\}$, and
(iii)~$\{f(q) \mid q \notin S\} = \{0\}$.
A~ranking is \emph{tight} if it is $Q$-tight; we use $\cT$ to denote the set of
all tight rankings.
For a~pair of rankings~$f$ and~$f'$, a~set $S \subseteq Q$, and a~symbol~$a \in
\Sigma$, we use $f' \rankleq f$ iff for every $q \in S$ and $q' \in \delta(q,
a)$ it holds that $f'(q') \leq f(q)$.

The \CompSchewe procedure constructs the BA~$\BSven = (Q', \delta', I', F')$
whose components are defined as follows:
%
%
%
%
\vspace{-1.5mm}
\begin{itemize}
  \item  $Q' = Q_1 \cup Q_2$ where
    \begin{itemize}
      \item  $Q_1 = 2^Q$ and
      \item  $Q_2 = \hspace*{-1mm}
        \begin{array}[t]{ll}
          \{(S,O,f, i) \in \hspace*{0mm}& 2^Q \times 2^Q \times \cT \times \{0, 2, \ldots, 2n - 2\} \mid {} \\
          &f \text{ is $S$-tight}, O \subseteq S \cap f^{-1}(i)\},
        \end{array}$
    \end{itemize}
  \item  $I' = \{I\}$,
  \item  $\delta' = \delta_1 \cup \delta_2 \cup \delta_3$ where
    \begin{itemize}
      \item  $\delta_1: Q_1 \times \Sigma \to 2^{Q_1}$ such that $\delta_1(S, a) =
        \{\delta(S,a)\}$,
      \item $\delta_2 : Q_1 \times \Sigma \to 2^{Q_2}$ such that $\delta_2(S, a) =
        \{(S', \emptyset, f, 0) \mid S' = \delta(S, a),\linebreak f \text{ is } S'\text{-tight}\}$, and
      \item $\delta_3: Q_2 \times \Sigma \to 2^{Q_2}$ such that $(S', O', f', i') \in
        \delta_3((S, O, f, i), a)$ iff
        $S' = \delta(S, a), f' \rankleq f$, $\rankof f = \rankof{f'}$, $f'$ is $S'$-tight, and
        \begin{itemize}
          \item  $i' = (i+2) \mod (\rankof{f'} + 1)$ and $O' = f'^{-1}(i')$ if
            $O = \emptyset$ or
          \item  $i' = i$ and $O' = \delta(O, a) \cap f'^{-1}(i)$ if $O \neq
            \emptyset$, and
        \end{itemize}
    \end{itemize}
  \item  $F' = \{\emptyset\} \cup ((2^Q \times \{\emptyset\} \times \cT \times
    \omega) \cap Q_2)$.
\vspace{-1.5mm}
\end{itemize}
%

\noindent
Intuitively, \CompSchewe is an extension of the classical subset construction
for determinization of finite automata.
In particular, $Q_1, \delta_1$, and $I_1$ constitute the deterministic finite
automaton obtained from~$\A$ using the subset construction.
The automaton can, however, nondeterministically guess a~point at which it will
make a~transition to a~\emph{macrostate} $(S, O, f, i)$ in the $Q_2$~part;
this guess corresponds to a level in the run DAG of the accepted word from which
the ranks of all levels form an $S$-tight ranking, where the $S$~component of
the macrostate is again a~subset from the subset construction.
In the
$Q_2$~part, $\BSven$~makes sure that in order for a~word to be accepted by~$\BSven$,
all runs of~$\A$ over the word need to touch an accepting state only finitely
many times.
This is ensured by the~$f$ component, which,
roughly speaking, maps states to ranks of corresponding vertices in the run DAG over the
given word.
The $O$~component is used for a~standard cut-point construction, and is used to
make sure that all runs that have reached an~accepting state in~$\A$ will
eventually leave it (this can happen for different runs at a~different point).
The~$S, O$, and $f$~components were already present in~\cite{KupfermanV01}.
The $i$~component was introduced by Schewe to improve the complexity of the
construction; it is used to cycle over phases, where in each phase we focus on
cut-points of a~different rank.
See~\cite{Schewe09} for a~more elaborate exposition.

\begin{proposition}[Corollary~3.3 in \cite{Schewe09}]
\label{prop:sven_correct}
$\langof \BSven = \overline{\langof \A}$.
\end{proposition}

\vspace{-5.0mm}
\section{Purging Macrostates with Incompatible Rankings}\label{sec:purging}
\vspace{-2.0mm}

Our first optimisation is based on removing from~$\BSven$ macrostates $(S, O, f,
i) \in Q_2$ whose level ranking~$f$ assigns some states of~$S$ an unnecessarily
high rank.
Intuitively, when $S$~contains a~state~$p$ and a~state~$q$ such that~$p$ is
(directly) simulated by~$q$, i.e.~$p \dirsimby q$, then $f(p)$ needs to be at
most~$f(q)$.
This is because in any word~$\word$ and its run DAG $\dagw$ in~$\A$, if $p$ and
$q$ are at the same level~$i$ of~$\dagw$, then the ranks of their vertices $v_p$
and $v_q$ at the given level are either both $\omega$ (when $\word \in \langof
\A$), or such that $\rankwof{v_p} \leq \rankwof{v_q}$ otherwise.
This is because, intuitively, the DAG rooted in~$v_p$ in~$\dagw$ is isomorphic to a~subgraph of the~DAG rooted in~$v_q$.

Formally, consider the following predicate on macrostates of~$\BSven$:
\begin{equation}
  \toremdirof{S, O, f, i} \quad \text{iff} \quad \exists p,q \in S: p\dirsimby q
  \land f(p) > f(q) .
\end{equation}
%
We modify \CompSchewe to purge macrostates that satisfy~$\toremdir$.
That is, we create a~new procedure \PurgeDi obtained from \CompSchewe
by modifying the definition of~$\BSven$ such that all occurrences of~$Q_2$ are
substituted by~$\Qtwodi$ and
\begin{equation}
  \Qtwodi = Q_2 \setminus \{(S,O,f,i) \in Q_2 \mid \toremdirof{S,O,f,i}\}.
\end{equation}
We denote the BA obtained from \PurgeDi as~$\BSvenRankdir$.
The following lemma, proved in \secref{sec:proofs_purging} states the
correctness of this construction.
\begin{restatable}{lemma}{lemmaOptRankDir}
  \label{lem:optrankdir}
  $\langof \BSvenRankdir = \langof \BSven$
\end{restatable}

\noindent
The following natural question arises: Is it possible to extend the purging
technique from direct simulation to other notions of simulation?
For \emph{fair} simulation, this cannot be done.
The reason is that, for a~pair of states~$p$ and~$q$ s.t.~$p \fairsimby q$, it
can happen that for a~word~$\beta \in \Sigma^\omega$, there can be a~trace from
$p$ over~$\beta$ that finitely many times touches an accepting
state (i.e., a~vertex of~$p$ in the corresponding run DAG can have any rank
between~$0$ and~$2n$), while all traces from~$q$ over~$\beta$ can completely avoid
touching any accepting state.
From the point of view of fair simulation, these are both unfair traces, and,
therefore, disregarded.

On the other hand, \emph{delayed} simulation---which is often much richer than
direct simulation---can be used, with a~small change.
Intuitively, the delayed simulation can be used because $p \delsimby q$
guarantees that on every level of trees in~$\dagw$ rooted in~$v_p$ and in~$v_q$ respectively,
the rank of the vertex $v_p$ is at most by one larger than
the rank of vertex $v_q$ (or by any number smaller).
Formally,
let~$\toremdel$ be the following predicate on macrostates of~$\BSven$:
\begin{equation}
  \toremdelof{S, O, f, i}\quad\text{iff}\quad \exists p,q \in S: p\delsimby q
  \land f(p) > \rond{f(q)},
\end{equation}
%
%
where $\rond{x}$ for $x \in \omega$ denotes the smallest even number greater or
equal to~$x$ and $\rond \omega = \omega$.
Similarly as above, we create a new procedure, called \PurgeDe, which
is obtained from \CompSchewe by modifying the definition of~$\BSven$ such that
all occurrences of~$Q_2$ are substituted by~$\Qtwode$ and
\begin{equation}
  \Qtwode = Q_2 \setminus \{(S,O,f,i) \in Q_2 \mid \toremdelof{S,O,f,i}\}.
\end{equation}
We denote the BA obtained from \PurgeDe as~$\BSvenRankdel$.
\vspace{-1mm}
\begin{restatable}{lemma}{lemmaOptRankDel}
  \label{lem:optrankdel}
  $\langof \BSvenRankdel = \langof \BSven$
\end{restatable}

\noindent
The use of $\rond{f(q)}$ in~$\toremdel$ results in the fact that the two purging
techniques are incomparable.
For instance, consider a~macrostate $(\{p,q\}, \emptyset, \{p \mapsto 2, q
\mapsto 1\}, 0)$ such that $p \dirsimby q$ and $p \delsimby q$.
Then the macrostate will be purged in \PurgeDi, but not in \PurgeDe.

The two techniques can, however, be easily combined into a~third procedure
\PurgeDiDe, when~$Q_2$ is substituted in \CompSchewe with $\Qtwodide$ defined as
\begin{equation}
  \Qtwodide = Q_2 \setminus \{(S,O,f,i) \in Q_2 \mid \toremdirof{S,O,f,i} \lor \toremdelof{S,O,f,i}\}.
\end{equation}
We denote the resulting BA as~$\BSvenRankdirdel$.

\vspace{-1mm}
\begin{restatable}{lemma}{lemmaOptRankDirdel}
	\label{lem:optrankdirdel}
  $\langof \BSvenRankdirdel = \langof \BSven$
\end{restatable}


\vspace{-5.0mm}
\subsection{Proofs of Lemmas~\ref{lem:optrankdir}, \ref{lem:optrankdel}, and \ref{lem:optrankdirdel}}
\label{sec:proofs_purging}
\vspace{-1.0mm}

\enlargethispage{3mm}

We first give a~lemma that an~$x$-strategy~$\strategx$ preserves an
$x$-simulation~$\xsimby$.

\vspace{-1mm}
\begin{lemma}\label{lem:strind}
	Let $\xsimby$ be an $x$-simulation (for $x \in \{\mathit{di, de, f}\}$). Then, the following holds:
	$\forall p,q \in Q: p\xsimby q \land p \ltr a p' \in \delta \limplies
	\exists q'\in Q: q \ltr a q' \in \delta \land p' \xsimby q'$.
\end{lemma}

\begin{proof}
  Let~$p,q \in Q$ such that~$p \xsimby q$ and $p \ltr a p' \in \delta$, and let
  $\pi_p$ be a trace starting from~$p$ with the first transition $p \ltr a p'$.
  From the definition of $x$-simulation, there is a winning Duplicator
  strategy~$\strategx$; let~$\pi_q = \strategx(q', \pi_p)$ and let $q \ltr a q'$
  be the first transition of~$\pi_q$.
  Let $\pi_{p'}$ and $\pi_{r'}$ be traces obtained from $\pi_p$ and $\pi_r$ by
  removing their first transitions.
  It is easy to see that if~$\condof x {\pi_{p}, \pi_{r}}$ then also $\condof x
  {\pi_{p'}, \pi_{r'}}$ for any $x \in \{\mathit{di, de, f}\}$.
  It follows that~$\strategx$ is also a~winning Duplicator strategy
  from~$(p',r')$.
  \qed
\end{proof}






\noindent
Next, we focus on delayed simulation and the proof of
Lemma~\ref{lem:optrankdel}.
In the next lemma, we show that if there is a~pair of vertices on some
level of the run DAG where one vertex delay-simulates the other one, there
exists a~relation between their rankings.
This will be used to purge some useless rankings from the complemented BA.

\vspace{-1mm}
\begin{lemma}\label{lem:simrundag}
  Let $p,q \in Q$ such that $p \delsimby q$
	and $\dagw = (V, E)$ be the run DAG of $\A$ over~$\word$.
  For all $i \geq 0$, it holds that $(p, i) \in V \wedge (q, i) \in V \limplies
  \rankwof{p, i} \leq \rond{\rankwof{q, i}}$.
\end{lemma}

\begin{proof}
  Consider some $(p,i)\in V$ and $(q,i)\in V$.
  First, suppose that~$\rankwof{q,i} = \omega$.
  Since the rank can be at most~$\omega$, it will always hold that $\rankwof{p, i}
  \leq \rond{\rankwof{q, i}}$.

  On the other hand, suppose that~$\rankwof{q, i}$ is finite, i.e.,
  $\wordof{i:\omega}$ is not accepted by~$q$.
  Then, due to Lemma~\ref{lem:rankings}, $0 \leq \rankwof{q, i} \leq 2n$.
  Because $p \delsimby q$, it holds that $\wordof{i:\omega}$ is also not
  accepted by~$p$, and therefore also~$0 \leq \rankwof{p, i} \leq 2n$.
  We now need to show that $0 \leq \rankwof{p, i} \leq \rond{\rankwof{q, i}} \leq
  2n$.

  \enlargethispage{3mm}

  Let $\{\dagw^k\}_{k=0}^{2n+1}$ be the sequence of run DAGs obtained from
  $\dagw$ in the ranking procedure from \secref{sec:rundag}. In the following
  text we use the abbreviation $v \in\dagw^m\setminus\dagw^n$ for $v \in\dagw^m
  \wedge v\notin \dagw^n$. Since the rank of a~node~$(r, j)$ is given as the
  number~$l$ s.t.~$(r, j) \in \dagw^l \setminus \dagw^{l+1}$, we will finish the
  proof of this lemma by proving the following claim:
  \vspace{-1mm}
  \begin{claim}
    Let $k$ and $l$ be s.t.~$(p, i) \in \dagw^k \setminus \dagw^{k+1}$
    and~$(q, i) \in \dagw^l \setminus \dagw^{l+1}$.
    Then $k \leq \rond{l}$.
  \end{claim}
  \vspace{-1mm}
  \begin{claimproofnoqed}
    We prove the claim by induction on~$l$.
    \begin{itemize}
  \vspace{-1mm}
      \item  Base case: ($l = 0$) Since we assume $\A$ is complete, no vertex
        in~$\dagw^0$ is finite.

        ($l = 1$) We prove that if $(q, i)$ is endangered in $\dagw^1$,
        then $(p, i)$ is endangered in $\dagw^1$ as well (so both
        would be removed in $\dagw^2$).
        For the sake of contradiction, assume that $(q,i)$ is endangered in
        $\dagw^1$ and $(p,i)$ is not.
        Therefore, since $\dagw^1$ contains no finite vertices,
        there is an infinite path~$\pi$ from $(p,i)$ s.t. $\pi$ contains at
        least one accepting state.
        In the following, we abuse notation and, given a~strategy~$\strategdel$
        and a~state $s \in Q$, use $\strategdel((s, i), \pi)$ to denote the path
        $(s_0, i)(s_1, i+1)(s_2, i+2) \ldots$ such that $s_0 = s$ and
        $\forall j \geq 0$, it holds that $s_{j+1} = \strategdel(s_j, r_{i+j}
        \ltr{\wordof{i+j}} r_{i+j+1})$ where $\pi_x = (r_x, x)$ for every $x
        \geq 0$.
        Since $p \delsimby q$, there is a~corresponding infinite path $\pi' =
        \strategdel((q,i), \pi)$ that also contains at least one accepting state.
        Therefore, $(q, i)$ is not endangered, a~contradiction to the
        assumption, so we conclude that $l=1 \limplies k=1$.
    \item Inductive step: We assume the claim holds for all $l < 2j$ and prove
        the inductive step for even and odd steps independently.

        ($l = 2j$)
        We prove that if $(q,i)$ is finite in~$\dagw^l$ (and therefore would
        be removed in~$\dagw^{l+1}$), then either $(p,i) \notin
        \dagw^l$, or $(p,i)$ is also finite in $\dagw^l$.
        For the sake of contradiction, we assume that $(q,i)$ is finite
        in~$\dagw^l$ and that $(p,i)$ is in~$\dagw^l$, but is not
        finite there (and, therefore, $k > l$).
        Since $(p,i)$ is not finite in~$\dagw^l$, there is an infinite path
        $\pi$ from $(p,i)$ in~$\dagw^l$.
        Because $p \delsimby q$, it follows that there is an infinite path
        $\pi' = \strategdel((q,i), \pi)$ in $\dagw^0$ ($\pi'$~is
        not in~$\dagw^l$ because $(q, i)$ is finite there).
        Using Lemma~\ref{lem:strind} (possibly multiple times) and the fact
        that $(q,i)$ is finite, we can find vertices $(p',x)$ in~$\pi$ and
        $(q',x)$ in~$\pi'$ s.t.~$p'\delsimby q'$ and $(q',x)$ is not in $\dagw^l$,
        therefore, $(q',x) \in \dagw^e \setminus \dagw^{e+1}$ for some $e < l$.
        Because $(p',x) \in \dagw^l$ and it is not finite ($\pi$ is infinite),
				it follows that $(p',x) \in
        \dagw^f \setminus \dagw^{f+1}$ for some $f > l$, and since
        $e < l < f$, we have that $f\not\leq e+1$, implying $f\not\leq\rond
        e$, which is in contradiction to the induction hypothesis.

        $(l = 2j+1)$
        We prove that if $(q,i)$ is endangered in~$\dagw^l$ (and therefore would
        be removed in~$\dagw^{l+1}$), then either $(p,i) \notin
        \dagw^l$, or $(p,i)$ is removed at the latest in $\dagw^{l+1}$.
        For the sake of contradiction, assume that $(q,i)$ is endangered
        in~$\dagw^l$ while $(p,i)$ is removed later than in $\dagw^{l+1}$.
        Therefore, since $\dagw^l$ contains no finite vertices (they were
        removed in the $(l-1)$-th step), there is an infinite path~$\pi$
        from $(p,i)$ s.t. $\pi$ contains at least one accepting state.
        Because $p \delsimby q$, there is a~corresponding path $\pi' =
        \strategdel((q,i), \pi)$ from $(q,i)$ in~$\dagw^0$ that also contains at
        least one accepting state and moreover $\pi' \notin \dagw^l$.
        Since $\pi'$ has an infinite number of
				states (and at least one accepting), not all states from $\pi'$ were removed
				in $\dagw^{l-1}$, i.e., there is at least one node with rank less or equal
				to $l-2$.
        Using Lemma~\ref{lem:strind} (also possibly multiple times) we can hence
        find states $(p',x)$ in~$\pi$ and $(q',x)$ in~$\pi'$ s.t.~$p'\delsimby q'$ and
        $(q',x)$ is not in $\dagw^l$ and has a~rank less or equal to $l-2$,
        therefore, $(q',x) \in \dagw^e \setminus \dagw^{e+1}$ for some $e <
        l-1$.
        Because $(p',x) \in \dagw^l$, it follows that $(p',x) \in
        \dagw^f \setminus \dagw^{f+1}$ for some $f \geq l$, and,
        therefore, $f\not\leq e + 1$, which is in contradiction to the induction
        hypothesis.
        \claimqed
    \end{itemize}
  \end{claimproofnoqed}
  This concludes the proof.
  \qed
\end{proof}

\vspace{-1mm}
\begin{lemma}\label{lem:simrundir}
  Let $p,q \in Q$ such that $p \dirsimby q$
	and $\dagw = (V, E)$ be the run DAG of $\A$ over~$\word$.
  For all $i \geq 0$, it holds that $(p, i) \in V \wedge (q, i) \in V \limplies
  \rankwof{p, i} \leq \rankwof{q, i}$.
\end{lemma}

\vspace{-1mm}
\begin{proof}
Can be obtained as a~simplified version of the proof of
Lemma~\ref{lem:simrundag}.
\qed
\end{proof}

\noindent
We are now ready to prove Lemma~\ref{lem:optrankdel}.

\vspace{-1mm}
\lemmaOptRankDel*

\vspace{-1mm}
\begin{proof}
	\begin{itemize}

    \item[$(\subseteq)$]  Follows directly from the fact that~$\BSvenRankdel$ is
      obtained by removing states from~$\BSven$.

    \item[$(\supseteq)$]
      Let $\word \in \langof \BSven$.
			As shown in the proof of Lemma~3.2 in~\cite{Schewe09}, there are two
			cases.
      The first case is when all vertices of $\dagof \word$ are finite, which we
      do not need to consider, since we assume complete automata.

			The other case is when $\dagof \word$ contains an infinite
			vertex.
			In this case, $\BSven$ contains an accepting run $$\rho = S_0 S_1 \ldots S_p
			(S_{p+1}, O_{p+1}, f_{p+1}, i_{p+1}) (S_{p+2}, O_{p+2}, f_{p+2}, i_{p+2})
			\ldots$$ with
      \begin{itemize}
        \item  $S_0 = I, O_{p+1} = \emptyset$, and $i_{p+1} = 0$,
        \item  $S_{j+1} = \trans(S_j, \wordof j)$ for all $j \in \omega$,
      \end{itemize}
      and, for all $j > p$,
      \begin{itemize}
        \item  $O_{j+1} = f^{-1}_{j+1}(i_{j+1})$ if $O_j = \emptyset$ or \\
              $O_{j+1} = \trans(O_j, \wordof j) \cap f^{-1}_{j+1}(i_{j+1})$
              if $O_j \neq \emptyset$, respectively,
        \item  $f_j$ is the $S_j$-tight level ranking that maps each~$q \in
              S_j$ to the rank of $(q,j) \in \dagw$,
        \item  $i_{j+1} = i_j$ if $O_j \neq \emptyset$ or \\
              $i_{j+1} = (i_j + 2) \mod (\rankof f + 1)$ if $O_j =
              \emptyset$, respectively.
      \end{itemize}
      The ranks assigned by~$f_j$ to states of~$S_j$ match the ranks of
      the corresponding vertices in~$\dagof \word$.

				$\circledast$ Using Lemma~\ref{lem:simrundag}, we conclude
				that $\rho$ contains no macrostate $(S,O,f,j)$ where $f(p)
				> \rond{f(q)}$ and $p\delsimby q$ for $p, q \in S$. Therefore, $\rho$~is
				also an accepting run in~$\BSvenRankdel$.
        (We use $\circledast$ to refer to this paragraph later.)
        \qed
	\end{itemize}
\end{proof}

\lemmaOptRankDir*
\begin{proof}
The same as for Lemma~\ref{lem:optrankdel} with $\circledast$ substituted by the
following:

\noindent
$\circledast$ Using Lemma~\ref{lem:simrundir}, we conclude that
$\rho$~contains no macrostate $(S,O,f,j)$ where $f(p) > f(q)$
and $p\dirsimby q$ for $p, q \in S$. So $\rho$~is also an
accepting run in~$\BSvenRankdir$.\qed
\end{proof}

\lemmaOptRankDirdel*
\begin{proof}
The same as for Lemma~\ref{lem:optrankdel} with $\circledast$ substituted by the
following:

\noindent
$\circledast$ Using Lemmas~\ref{lem:simrundir}
and~\ref{lem:simrundag}, we conclude that $\rho$~contains no macrostate
$(S,O,f,j)$ where either $f(p) > f(q)$ and $p\dirsimby q$, or $f(p) >
\rond{f(q)}$ and $p\delsimby q$ for $p, q \in S$. Therefore, $\rho$~is also
an accepting run in~$\BSvenRankdirdel$.
\qed
\end{proof}

\vspace{-3.0mm}
\section{Saturation of Macrostates}\label{sec:saturation}
\vspace{-2.0mm}

Our second optimisation is inspired by an optimisation of determinisation of
classical finite automata from~\cite[Section~5]{GlabbeekP08}.
Their optimisation is based on saturating every constructed macrostate in the
classical subset construction with all direct-simulation-smaller states.
This can reduce the total number of states of the determinized automaton because
two or more macrostates can be mapped to a~single saturated macrostate.
(In Section~\ref{sec:compression}, we show why an analogue of their
\emph{compression} cannot be used.)

We show that a~similar technique can be applied to BAs.
We do not restrain ourselves to direct simulation, though, and generalize the
technique to delayed simulation.
In particular, in our optimisation, we saturate the $S$~components of
macrostates $(S, O, f, i)$ obtained in \CompSchewe with all $\delsimby$-smaller
states.
Formally, we modify \CompSchewe by substituting the definition of the
constructed transition function~$\delta'$ with~$\deltasatdown'$ defined as
follows:
\begin{itemize}
  \item  $\deltasatdown' = \deltasatup_1 \cup \deltasatup_2 \cup \deltasatup_3$
    where
    \begin{itemize}
      \item  $\deltasatup_1: Q_1 \times \Sigma \to 2^{Q_1}$ with
        $\deltasatup_1(S, a) = \{\str{\delta(S,a)}\}$,
      \item $\deltasatup_2 : Q_1 \times \Sigma \to 2^{Q_2}$ with
        $\deltasatup_2(S, a) = \{(S', \emptyset, f, 0) \mid S' = \str{\delta(S,
        a)}\}$, and
      \item $\deltasatup_3: Q_2 \times \Sigma \to 2^{Q_2}$ with $(S', O', f',
        i') \in \deltasatup_3((S, O, f, i), a)$ iff
        $S' = \str{\delta(S, a)}, f' \rankleq~f$, $\rankof f = \rankof{f'}$, and
        \begin{itemize}
          \item  $i' = (i+2) \mod (\rankof{f'} + 1)$ and $O' = f'^{-1}(i')$ if
            $O = \emptyset$ or
          \item  $i' = i$ and $O' = \delta(O, a) \cap f'^{-1}(i)$ if $O \neq
            \emptyset$,
        \end{itemize}
    \end{itemize}
\end{itemize}
where $\str{S} = \{ q\in Q ~|~\exists s \in S: q\delsimby s \}$.
We denote the obtained procedure as \Saturate and the obtained BA
as~$\BSvenSat$.

\begin{restatable}{lemma}{lemmaSaturation}
  \label{lem:saturation}
  $\langof{\BSvenSat} = \langof{\BSven}$
\end{restatable}

\noindent
Obviously, as direct simulation is stronger than delayed simulation, the
previous technique can also use direct simulation only (e.g., when computing
the full delayed simulation is computationally too demanding).
Moreover, \Saturate is also compatible with all \PurgeX algorithms for $x \in
\{\mathit{di}, \mathit{de}, \mathit{di+de}\}$ (because they just remove
macrostates with incompatible rankings from~$Q_2$)---we call the combined
versions \PurgeX\!\!\!+\Saturate and the complement BAs they
output~$\BSvenSatX$.

\begin{restatable}{lemma}{propositionSatPurge}
	\label{lem:satpurge}
  $\langof \BSvenSatdi = \langof \BSvenSatde = \langof \BSvenSatdide = \langof
  \BSven$
\end{restatable}

\vspace{-0.0mm}
\subsection{Proofs of Lemmas~\ref{lem:saturation} and~\ref{lem:satpurge}}\label{sec:label}
\vspace{-0.0mm}

We start with a~lemma, used later, that talks about languages of states related
by delayed simulation when there is a~path between them.

\begin{lemma}\label{lem:sim-lang}
  For $p,q \in Q$ such that $p \delsimby q$, let $L_\top = \{ w\in\Sigma^*~|~
  p\ftransover w q\}$ and $L_\bot = \{ w\in\Sigma^*~|~ p\transover w
  q\}$.
  Then $L(q) \supseteq (L_\bot^*L_\top)^\omega$.
\end{lemma}

\begin{proof}
  First we prove the following claim:
  \begin{claim}
  For every word $\alpha = w_0w_1w_2\dots \in \Sigma^\omega$ where $w_i\in
  L_\top \cup L_\bot$, we can construct a~trace $\pi = p \transover{w_0} q_0
  \transover{w_1} q_1 \transover{w_2} \cdots$ over~$\alpha$ such that
  $p\delsimby q_0$ and $q_i \delsimby q_{i+1}$ for all $i \geq 0$.
  \end{claim}
  \begin{claimproofnoqed}
  We assign $q_0 := q$ and construct the rest of~$\pi$ by the following
    inductive construction.
  \begin{itemize}
    \item  Base case: ($i = 0$)
      From the assumption it holds that $p \transover{w_1} q_0$ and $p\delsimby
      q_0$.
      From Lemma~\ref{lem:strind} there is some $r \in Q$ s.t.~$q_0
      \transover{w_1} r$ and $q_0 \delsimby r$.
      We assign $q_1 := r$, so $q_0 \delsimby q_1$.

    \item  Inductive step:
      Let $\pi' = p \transover{w_0} q_0 \transover{w_1}
      \cdots \transover{w_i} q_i$ be a~prefix of a~trace such that $q_j
      \delsimby q_{j+1}$ for every $j < i$.
      From the transitivity of~$\delsimby$, it follows that $p \delsimby q_i$.
      From Lemma~\ref{lem:strind} there is some $r \in Q$ s.t.~$q_i
      \transover{w_i} r$ and $q \delsimby r$.
      We assign $q_{i+1} := r$, so $q_i \delsimby q_{i+1}$.
      \claimqed
  \end{itemize}
  \end{claimproofnoqed}
  %
  %
  %
  Consider a~word $\word \in (L_\bot^*L_\top)^\omega$ such that $\word =
  w_0w_1w_2\dots$ for $w_i\in L_\top\cup L_\bot$.
  We show that $\word \in \langof q$.
  According to the previous claim, we can construct a~trace $\pi = p
  \transover{w_0} q = q_0 \transover{w_1} q_1 \transover{w_2} \cdots$ over
  $\word$ s.t.~$p\delsimby q_0$ and $q_i \delsimby q_{i+1}$ for all $i \geq 0$.
  Since $p\delsimby q$, from Lemma~\ref{lem:strind} it follows that we can
  construct a~trace $\pi' = q \transover{w_0} r_0 \transover{w_1} r_1
  \transover{w_2} \cdots$ s.t.~$q_i \delsimby r_{i}$ for every $i \geq 0$.
  Because $\word$ contains infinitely often a subword from~$L_\top$, there is
  some $\ell \in \omega$ such that $q_\ell\transover{w_\ell}q_{\ell+1}$ and
  $r_\ell\transover{w_\ell}r_{\ell+1}$ for $w_\ell\in L_\top$.
  Note that it holds that $p\delsimby q_\ell \delsimby r_\ell$.
  We can again use the claim above to construct a~trace $\pi^\star = p
  \ftransover{w_\ell} q = s_0 \transover{w_{\ell+1}} s_1 \transover{w_{\ell+2}}
  \cdots$ over $\word_\ell = w_\ell w_{\ell+1} w_{\ell+2} \dots$ such that $p\delsimby s_0$
  and $s_i \delsimby s_{i+1}$ for all $i \geq 0$.
  Since $p \delsimby r_\ell$, we can simulate $\pi^\star$ from $r_\ell$ by
  a~trace $\pi^{\star\prime}$, and because $p \ftransover{w_\ell} q$, we know
  that $\pi^{\star\prime}$ will touch an accepting state in finitely many steps
  (this holds because~$w_\ell$ is from~$L_\top$, which are the words over which we can go
  from~$p$ to~$q$ and touch an accepting state).
  Consider $m \geq \ell$ such that~$s_m$ is the first state after the accepting
  state that is one of the $\{s_0, s_1, \ldots\}$ in~$\pi^{\star\prime}$.
  This reasoning could be repeated for all occurrences of a~subword
  from~$L_\top$ in~$\pi^\star$, therefore~$\word \in \langof q$.
  \qed
\end{proof}

\noindent
Next, we give a~lemma used for establishing correctness of saturating
macrostates with $\delsimby$-smaller states.

\begin{lemma}
  \label{lem:saturation_preserves_lang}
  Let $p,q,r \in Q$ such that $r \ltr a q \in \delta$ and $p\delsimby q$.
  Further, let $\A' = (Q,\delta',I,F)$ where $\delta' = \delta \cup \{ r \ltr{a} p \}$.
  Then $\langof{\A} = \langof{\A'}$.
\end{lemma}
\begin{proof}
\begin{itemize}
  \item[$(\subseteq)$] Clear.

  \item[$(\supseteq)$] Consider some $\word \in \langof{\A'}$ and an accepting
    trace~$\pi$ in~$\A'$ over~$\word$.
    There are two cases:
  \begin{enumerate}
    \item ($\pi$ contains only finitely many transitions $r \ltr{a} p$)\\
      In this case, $\pi$ is of the form $\pi = \pi_i \pi_\omega$ where
      $\pi_i$~is a~finite prefix
      $\pi_i = q_0 \transover{w_0} r \ltr{a} p \transover{w_1} r \ltr a p
      \transover{w_2} \cdots \transover{w_n} r \ltr{a} p$, for $q_0 \in I$,
      and $\pi_\omega$ is an infinite trace from~$p$ that does not contain any
      occurrence of the transition $r \ltr a p$.
      We construct in~$\A$ a~trace $\pi' =
      q_0 \transover{w_0} r \ltr{a} q \transover{w_1} r_1 \ltr a q_1
      \transover{w_2}\cdots \transover{w_n} r_n \ltr{a} q_n . \pi'_\omega$
      as follows.
      Let $\strategdel$ be a~strategy for~$\delsimby$.
      We set $r_1 := \strategdel(q, p \transover{w_1} r)$, so $r \delsimby r_1$.
      Since $r \ltr a q \in \delta$, it follows that there is $r_1 \ltr a q_1
      \in \delta$ such that $p \delsimby q_1$.
      For $i > 1$, we set $r_i := \strategdel(q_{i-1}, p \transover{w_i} r)$.
      By induction, it follows that $\forall 1 \leq i \leq n: p \delsimby q_i$,
      in particular $p \delsimby q_n$.
      We set $\pi'_\omega := \strategdel(q_n, \pi_\omega)$.
      Since $\pi_\omega$ starts in~$p$ and contains infinitely many accepting states
      and $\pi'_\omega$ starts in $q_n$ and $p \delsimby q_n$, then
      $\pi'_\omega$ also contains infinitely many accepting states.
      It follows that $\pi'$ is accepting, so~$\word \in \langof \A$.

    \item ($\pi$ contains infinitely many transitions $r \ltr{a} p$)\\
      In this case, $\pi$ is of the form $\pi = q_0 \transover{w_0} r \ltr{a} p
      \transover{w_1} r \ltr a p \transover{w_2} \cdots \transover{w_n} r
      \ltr{a} p \transover{w_\omega} \cdots$, for $q_0 \in I$ and
      $\word = w_0 a w_1 a w_2 \dots$
      Since $\pi$ is accepting, for infinitely many $i\in\omega$, we have
      $p\ftransover{w_i a} p$ in $\A'$ and hence also $p\ftransover{w_ia} q$ in
      the original BA~$\A$.
      Using Lemma~\ref{lem:sim-lang} and the fact that $p \delsimby q$, we have
      $w_1aw_2a\dots \in L(q)$ and hence $\word = w_0 a w_1 a w_2 a
      \dots\in\langof{\A}$.
      \qed
  \end{enumerate}
\end{itemize}

\end{proof}

\vspace{-1mm}
\noindent
The following lemma guarantees that adding transitions in the way of
Lemma~\ref{lem:saturation_preserves_lang} does not break the computed delayed
simulation and can, therefore, be performed repeatedly, without the need to
recompute the simulation.

\enlargethispage{3mm}

\vspace{-1mm}
\begin{lemma}
  \label{lem:saturation_preserves_sim}
  Let $\delsimby$ be the delayed simulation on~$\A$.
  Further, let $p,q,r \in Q$ be such that $r \ltr{a} q \in \delta$ and
  $p\delsimby q$, and let $\A' = (Q,\delta',I,F)$ where $\delta' = \delta \cup
  \{ r \ltr{a} p \}$.
  Then $\delsimby$ is included in the delayed simulation on~$\A'$.
\end{lemma}
\vspace{-2mm}
\begin{proof}
  Let~$\strategdel$ be a~dominating strategy compatible with~$\delsimby$ and $\strategdel'$
  be a~strategy defined for all~$s \in Q$ such that $r \delsimby s$
  as~$\strategdel'(s, x) = \strategdel(s, x)$ when $x \neq (r \ltr a p)$ and
  $\strategdel'(s, r \ltr a p) = \strategdel(s, r \ltr a q)$.
  Note that $\strategdel'$ is also dominating wrt~$\delsimby$.
  This can be shown by the following proof by contradiction:
  Suppose~$\strategdel'$ is not dominating; then there is a~strategy~$\rho$
  such that~$\strategdel'(s, r \ltr{a} p)$ must be simulated by~$\rho(s, r
  \ltr{a} p) = t$.
  But then $\strategdel(s, r \ltr{a} q)$ must also (transitivity of simulation)
  be simulated by $t$, so $\strategdel$ is not dominating.
  Contradiction.

  Further, let $t,u \in Q$ be such that~$t \delsimby u$.
  %
  Let $\pi_t = t \transover{w_1} t_f \transover{w_2} r \ltr {a} p.
  \pi'_t$ be a~trace over $\word = w_1 w_2 a w_\omega \in \Sigma^\omega$
  in~$\A'$ such that~$t_f$ is an accepting state and $t_f \transover{w_2} r$
  does not contain any occurrence of~$r \ltr a p$.
  Further, let $\pi_u = u_0 \transover{w_1} u_f \transover{w_2} u_i \ltr {a}
  u_{i+1}.\pi'_u$ be a~trace corresponding to a~run~$u_0 u_1 u_2 \dots$ over
  $\word$ in~$\A$, where~$u_0 = u$, constructed as $\pi_u = \strategdel'(u,
  \pi_t)$.

\vspace{-1mm}
  \begin{claim}
    There is a~trace $\pi_v = t \transover{w_1} v_f . \pi_v'$
    over~$\word$ such that $\pi'_v$ contains an accepting state and $\pi_v$ is
    $\delsimby$-simulated by~$\pi_u$ at every position.
  \end{claim}
\vspace{-1mm}
  \begin{claimproofnoqed}
    %
    %
    We have the following two cases:
    \begin{itemize}
      \item  ($t \transover{w_1} t_f$ does not contain any occurrence of $r
        \ltr a p$)\\
    Let~$\pi_v = t \transover{w_1} t_f \transover{w_2} r \ltr a q . \pi'_v$
    be a~trace in~$\A$ over~$\word$ obtained from~$\pi_t$ by starting with its
    prefix up to~$r$, taking~$r \ltr a q$, and continuing with~$\pi'_v =
    \strategdel'(q, \pi'_t)$.
    Since in~$\pi_v$, it holds that $t_f$ is at the same position as~$t_f$
    in~$\pi_t$, the first part of the claim holds.
    Further, $\pi_u$~clearly $\delsimby$-simulates $\pi_v$ on $t \transover{w_1}
    t_f \transover{w_2} r$, and because~$\strategdel'$ simulates $r \ltr a p$ by
    a~transition to a~state~$u_{i+1}$ such that $q \delsimby u_{i+1}$
    and~$\pi'_v$ is constructed using~$\strategdel'$, then also the second part
    of the claim holds.


    %
    \item  ($t \transover{w_1} t_f$ contains at least one occurrence of $r
      \ltr a p$)\\
      Suppose that $\pi_t$ starts with $t \transover{w_{11}} r
      \ltr a p \transover{w_{12}} t_f$ such that $t \transover{w_{11}} r$ does
      not contain any $r \ltr a p$.
      Then let us start building~$\pi_v$ such that it starts with $t
      \transover{w_{11}} r \ltr a q$.
      On this prefix, $\pi_v$ is clearly $\delsimby$-simulated by the
      corresponding prefix of~$\pi_u$.
      We continue from~$q$ using the strategy~$\strategdel'$.
      In particular, the next time we reach $r \ltr a p$ in~$\pi_t$ while we are
      at some state~$v_1$ such that $r \delsimby v_1$, we simulate the
      transition by $\strategdel'(v_1, r \ltr a p)$ and so on.
      We can observe that when we arrive to~$t_f$ in~$\pi_t$, we also arrive
      to~$v_f$ in $\pi_v$ such that $t_f \delsimby v_f$.
      Therefore, $\pi'_v$ contains an accepting state.
      Moreover, since $\strategdel'$ is dominating, the second part of the claim
      also holds.
    \claimqed
    \end{itemize}
  \end{claimproofnoqed}
  \vspace{-2mm}
  From the claim above, it follows that
  the trace $u_f \transover{w_2} u_i \ltr a u_{i+1}.\pi'_u$ contains an
  accepting state, so~$\condof{\mathit{de}}{\pi_t, \pi_u}$.
  \qed
  %
\end{proof}

  \vspace{-1mm}
  \enlargethispage{3mm}
  
%
%
%

\noindent
Finally, we are ready to prove Lemma~\ref{lem:saturation}.

\vspace{-1mm}
\lemmaSaturation*

\vspace{-2mm}
\begin{proof}
  \begin{itemize}
    \item[$(\subseteq)$]
      Let $\word\in \langof{\BSvenSat}$ and~$\rho$ be an arbitrary accepting run
      over~$\word$ in~$\BSvenSat$ such that $\rho = S_0S_1 \dots S_{n-1}(S_{n},
      O_{n}, f_{n}, i_{n})(S_{n+1}, O_{n+1}, f_{n+1}, i_{n+1}) \dots$.
      For the sake of contradiction, assume that $\word\in \langof{\A}$,
      therefore, there is a~run~$\rho'$ on~$\word$ in~$\A$ having infinitely
      many accepting states.
      From the fact that tight level rankings form a non-increasing sequence, we have that
      $f_n(\rho'(n)) \geq f_{n+1}(\rho'(n+1)) \geq \cdots$.
      This sequence eventually stabilizes and from the property of level
      rankings and the fact that~$\rho'$ is accepting, it stabilizes in
      some~$\ell$ such that $f_\ell(\rho'(\ell))$ is even.
      This, however, means that the $O$~component of macrostates in~$\rho$
      cannot be emptied infinitely often, and, therefore, $\rho$~is not
      accepting, which is a~contradiction.
      Hence $\word\notin \langof{\A}$, so (from
      Proposition~\ref{prop:sven_correct}) $\word\in \langof{\BSven}$.

      %
      %

    \item[$(\supseteq)$]
      Consider some $\word\in \langof{\BSven}$.
      Let~$\A'$ be a~BA obtained from~$\A$ by adding transitions according
      to~Lemma~\ref{lem:saturation_preserves_sim}.
      Then from Lemma~\ref{lem:saturation_preserves_lang}, we have that
      $\langof{\A} = \langof{\A'}$.
      Therefore, $\word\in\langof{\BSven'}$ where $\BSven'$ is the BA obtained
      from~$\A'$ using \CompSchewe.
      It is easy to see that we can construct a~run in~$\BSvenSat$ that mimics
      the levels of run DAG of~$\word$ in~$\A'$ (i.e., we are able to empty the
      $O$~component infinitely often).
      Hence $\word\in\langof{\BSvenSat}$.
      \qed
%
	\end{itemize}
\end{proof}


\vspace{-2mm}
\propositionSatPurge*
\vspace{-2mm}

\begin{proof}
  \begin{itemize}
    \item[$(\subseteq)$]
      This part is the same as in the proof of Lemma~\ref{lem:saturation}.

    \item[$(\supseteq)$]
      Consider some $\word\in \langof{\BSven}$.
      Let~$\A'$ be a~BA obtained from~$\A$ by adding transitions according
      to~Lemma~\ref{lem:saturation_preserves_sim}.
      Then from Lemma~\ref{lem:saturation_preserves_lang}, we have that
      $\langof{\A} = \langof{\A'}$.
      Therefore, $\word\in\langof{\BSven'}$ where $\BSven'$ is the BA obtained
      from~$\A'$ using \CompSchewe.
      It is easy to see that we can construct a~run in~$\BSvenSat$ that mimics
      the levels of run DAG of~$\word$ in~$\A'$ (i.e., we are able to empty the
      $O$~component infinitely often). Using Lemmas~\ref{lem:simrundir}
			and~\ref{lem:simrundag}, we can conclude that the run contains no macrostate
			of the form $(S,O,f,j)$, where $f(p) > f(q)$ and $p\dirsimby q$, or $f(p) >
			\rond{f(q)}$ and $p\delsimby q$ for $p, q \in S$. Therefore, $\rho$ is also
			an accepting run in~$\BSvenSatdide$.
      Hence $\word\in\langof{\BSvenSatdide}$.
      \qed
	\end{itemize}
\end{proof}

\vspace{-0.0mm}
\subsection{Remarks on Compression of Macrostates}\label{sec:compression}
\vspace{-0.0mm}

An analogy to saturation of macrostates is their
compression~\cite[Section~6]{GlabbeekP08},
based on removing simulation-smaller states from a macrostate. This is,
however, not possible even for direct simulation, as we can see in the following
example.

\begin{example}
Consider the BA over $\Sigma=\{a\}$ given below.

\begin{center}
  \begin{tikzpicture}[shorten >=1pt,node distance=2cm,on grid,auto]
     \node[state,initial below, initial text={}] (p)   {$p$};
     \node[state,accepting] (q) [left=of p] {$q$};
     \node[state,accepting] (r) [right=of p] {$r$};
      \path[->]
      (p) edge  node[above] {$a$} (q)
          edge  node[above] {$a$} (r)
          edge [loop above] node {$a$} ()
      (q) edge [loop above] node {$a$} ()
      (r) edge [loop above] node {$a$} ();
  \end{tikzpicture}
\end{center}
For this BA we have $q\dirsimby r$ and $r\dirsimby q$. If we compress the
macrostates obtained in \CompSchewe, there is the following trace in the
output automaton:
\begin{align*}
  \{p\} &\ltr{a} (\{ p,q \}, \emptyset, \{ p\mapsto 3, q\mapsto 2, r\mapsto 1 \}, 0) \ltr{a}
  (\{ p,r \}, \{ r \}, \{ p\mapsto 3, q\mapsto 1, r\mapsto 2 \}, 2)\\
  &\ltr{a} (\{ p,q \}, \emptyset, \{ p\mapsto 3, q\mapsto 2, r\mapsto 1 \}, 2)
  \ltr{a} (\{ p,r \}, \{ r \}, \{ p\mapsto 3, q\mapsto 1, r\mapsto 2 \}, 0)\\
  &\ltr{a} (\{ p,q \}, \emptyset, \{ p\mapsto 3, q\mapsto 2, r\mapsto 1 \}, 0)\ltr{a} \cdots
\end{align*}
This trace contains infinitely many final states (we flush the $O$-set
infinitely often), hence we are able to accept the word~$a^\omega$, which is,
however, in the language of the input BA.
\qed
\end{example}

\vspace{-0.0mm}
\section{Use after Simulation Quotienting}\label{sec:label}
\vspace{-0.0mm}

In this short section, we establish that our optimizations introduced in
Sections~\ref{sec:purging} and~\ref{sec:saturation} can be applied with no
additional cost in the setting when BA complementation is preceded with
simulation-based reduction of the input BA (which is usually helpful), i.e.,
when the simulation is already computed beforehand for another purpose.
In particular, we show that simulation-based reduction preserves the
simulation (when naturally extended to the quotient automaton).
First, let us formally define the operation of quotienting.

Given an $x$-simulation $\xsimby$ for $x \in \{\mathit{di}, \mathit{de}\}$, we
use $\xsimilar$ to denote the $x$-\emph{similarity} relation (i.e., the
symmetric fragment) ${\xsimilar} = {\xsimby} \cap {\simby_x^{-1}}$.
Note that since $\xsimby$ is a~preorder, it holds that~$\xsimilar$ is an
equivalence.
We use $\clxof q$ to denote the equivalence class of~$q$ wrt~$\xsimilar$.
The \emph{quotient} of a~BA~$\A = (Q, \delta, I, F)$ wrt~$\xsimilar$ is the
automaton
\begin{equation}
  \A/{\xsimilar} = (Q/{\xsimilar}, \delta_{\xsimilar}, I_{\xsimilar},
  F_{\xsimilar})
\end{equation}
with the transition function $\delta_{\xsimilar}(\clxof q, a) = \{\clxof r \mid r \in \delta(\clxof q, a) \}$ and the set of initial and accepting states
$I_{\xsimilar} = \{ \clxof{q}\in Q/{\xsimilar} \mid q \in I\}$ and
$F_{\xsimilar} = \{ \clxof{q}\in Q/{\xsimilar} \mid q \in F\}$ respectively.
\begin{proposition}[\cite{bustan2003simulation}, \cite{EtessamiWS05}]
If $x \in \{\mathit{di}, \mathit{de}\}$, then $\langof{\A/{\xsimilar}} = \langof
\A$.
\end{proposition}

\begin{remark}[\cite{EtessamiWS05}]
$\langof{\A/{\fairsimilar}} \neq \langof \A$
\end{remark}

%
\noindent
Finally, the following lemma shows that quotienting preserves direct and delayed
simulations, therefore, when complementing~$\A$, it is possible to first
quotient~$\A$ wrt a~direct/delayed simulation and then use the same simulation
(lifted to the states of the quotient automaton) to optimize the
complementation.

\begin{restatable}{lemma}{lemmaQuotientPreservesSim}
  \label{lem:quotientPreservesSim}
Let $\xsimby$ be the $x$-simulation on~$\A$ for $x \in \{\mathit{di},
\mathit{de}\}$.
Then the relation $\quotxsimby$ defined as $\clxof q \quotxsimby \clxof r$
iff $q \xsimby r$ is the $x$-simulation on~$\A/{\xsimilar}$.
\end{restatable}

\begin{proof}
First, we show that $\quotxsimby$ is well defined, i.e., if $q \xsimby r$, then
for all~$q' \in \clxof q$ and $r' \in \clxof r$, it holds that $q' \xsimby r'$.
Indeed, this holds because $q' \xsimilar q$ and $r \xsimilar r$, and therefore
$q' \xsimby q \xsimby r \xsimby r'$; the transitivity of simulation yields $q'
\xsimby r'$.

Next, let~$\strategx$ be a~strategy that gives $\xsimby$.
Consider a trace defined as $\clxof{\pi_q} = \clxof{q_0} \ltr{\wordof 0} \clxof{q_1}
\ltr{\wordof  1} \cdots$ over a~word~$\word \in \Sigma^\omega$ in
$\A/{\xsimilar}$. Then,
\begin{enumerate}
  \item for $x = \mathit{di}$ there is a~trace $\pi_q = q_0' \ltr{\wordof 0}
  q_1' \ltr{\wordof 1} \cdots$ in~$\A$ s.t. $q_0'\in\clxof{q_0}$ and $q_i
  \xsimby q_i'$ for $i \geq 0$. Therefore, if $\clxof{q_i}$ is accepting then so is
  $q_i'$;
  \item for $x = \mathit{de}$ there is a~trace $\pi_q = q_0'
  \ltr{\wordof 0} q_1' \ltr{\wordof 1} \cdots$ in~$\A$ s.t.
  $q_0'\in\clxof{q_0}$, $q_i \xsimby q_i'$ for $i \geq 0$ and, moreover, if
  $\clxof{q_i}$ is accepting then there is $q_k'$ for $k \geq i$ s.t. $q_k'\in F$.
\end{enumerate}

\noindent
Further, let $\clxof{q_0} \quotxsimby \clxof{r_0}$. Then there is a trace $\pi_r
= \strategx(r, \pi_q) = (r = r_0) \ltr{\wordof 0} r_1 \ltr{\wordof 1} \cdots$
simulating~$\pi_q$ in~$\A$ from~$r$.
Further, consider its projection $\clxof{\pi_r} = \clxof{r_0}
\ltr{\wordof 0} \clxof{r_1} \ltr{\wordof  1} \cdots$
into~$\A/{\xsimilar}$. For all $i \geq 0$, we have that $q_i \xsimby r_i$, and
therefore also $\clxof{q_i} \quotxsimby \clxof{r_i}$. Since~$\condof x {\pi_q,
\pi_r}$, then also $\condof x {\clxof{\pi_q}, \clxof{\pi_r}}$.

Finally, we show that $\quotxsimby$ is maximal. For the sake of contradiction,
suppose that $\clxof r$ is $x$-simulating $\clxof q$ for some $q, r \in Q$
s.t.~$q \not\xsimby r$. Consider a word $\word \in \Sigma^\omega$
and a trace $\pi_q = (q = q_0) \ltr{\wordof 0} q_1 \ltr{\wordof 1} \cdots$ over
$\word$ in $\A$. Then there is a trace $\clxof{\pi_q} = \clxof{q = q_0}
\ltr{\wordof 0} \clxof{q_1} \ltr{\wordof  1} \cdots$ over~$\word$ in
$\A/{\xsimilar}$. According to the assumption, there is also a~trace
$\clxof{\pi_r} = \clxof{r = r_0} \ltr{\wordof 0} \clxof{r_1} \ltr{\wordof  1}
\cdots$ such that $\clxof{\pi_r}$ is $x$-simulating $\clxof{\pi_q}$. But then
there will also exist a trace $\pi_r = (r = r_0) \ltr{\wordof 0} r_1'
\ltr{\wordof 1} r_1' \ltr{\wordof 2} \cdots$ such that $r_i \xsimby r_i'$ for
all $i \in \omega$ and
$\condof x {\pi_q, \pi_r}$ (see the previous part of the proof).
Therefore, since $\xsimby$ is maximal, we have that $q \xsimby r$, which is in
contradiction with the assumption.
%
%
%
\qed
\end{proof}

\vspace{-0.0mm}
\section{Experimental Evaluation}\label{sec:label}
\vspace{-0.0mm}

We implemented our optimisations in a~prototype tool%
\footnote{%
\blinded{\url{https://github.com/vhavlena/ba-complement}}%
}
written in Haskell and performed preliminary experimental evaluation on a~set
of~124 random BAs with a non-trivial language over a two-symbol alphabet
generated using Tabakov and Vardi's model~\cite{TabakovV05}.
The parameters of input automata were set to the following bounds:
number of states: 6--7,
transition density: 1.2--1.3, and
acceptance density: 0.35--0.5.
Before complementing, the BAs were quotiented wrt the direct simulation for
experiments with \PurgeDi and the delayed simulation for experiments with \PurgeDe
and \PurgeDiDe.
The timeout was set to 300\,s.

We present the results for our strongest optimizations for \emph{outputs} of the size
up to 500 states in Fig.~\ref{fig:res-del-sim}.
As can be seen in Fig.~\ref{fig:purge}, purging alone often significantly
reduces the size of the output.
The situation with saturation is, on the other hand, more complicated.
In Fig.~\ref{fig:sat-del}, we can see that in some cases, the saturation
produces even smaller BAs than only purging, on the other hand, in some cases,
larger BAs are produced.
This is expected, because saturating the $S$~component of macrostates also means
that more level rankings (the $f$~component) need to be considered.

For outputs of a~larger size (we had 11 of them), the results follow a~similar
trend, but the probability that saturation will increase the size of the result
decreases.
For some concrete results, for one BA, the size of the output BA decreased from
4065 (\CompSchewe) to 985 (\PurgeDiDe) to 929 (\PurgeDiDe\!\!\!+\Saturate),
which yields a reduction to 24\,\%, resp. 22\,\%!
Further, we observed that all \PurgeX methods usually give similar results,
with the difference of only a~few states (when \PurgeDi and \PurgeDe differ,
\PurgeDi usually wins over \PurgeDe).


\begin{figure}[t]
  \centering
  \begin{subfigure}[b]{0.48\textwidth}
      \hspace*{-8mm}
      \includegraphics[width=70mm]{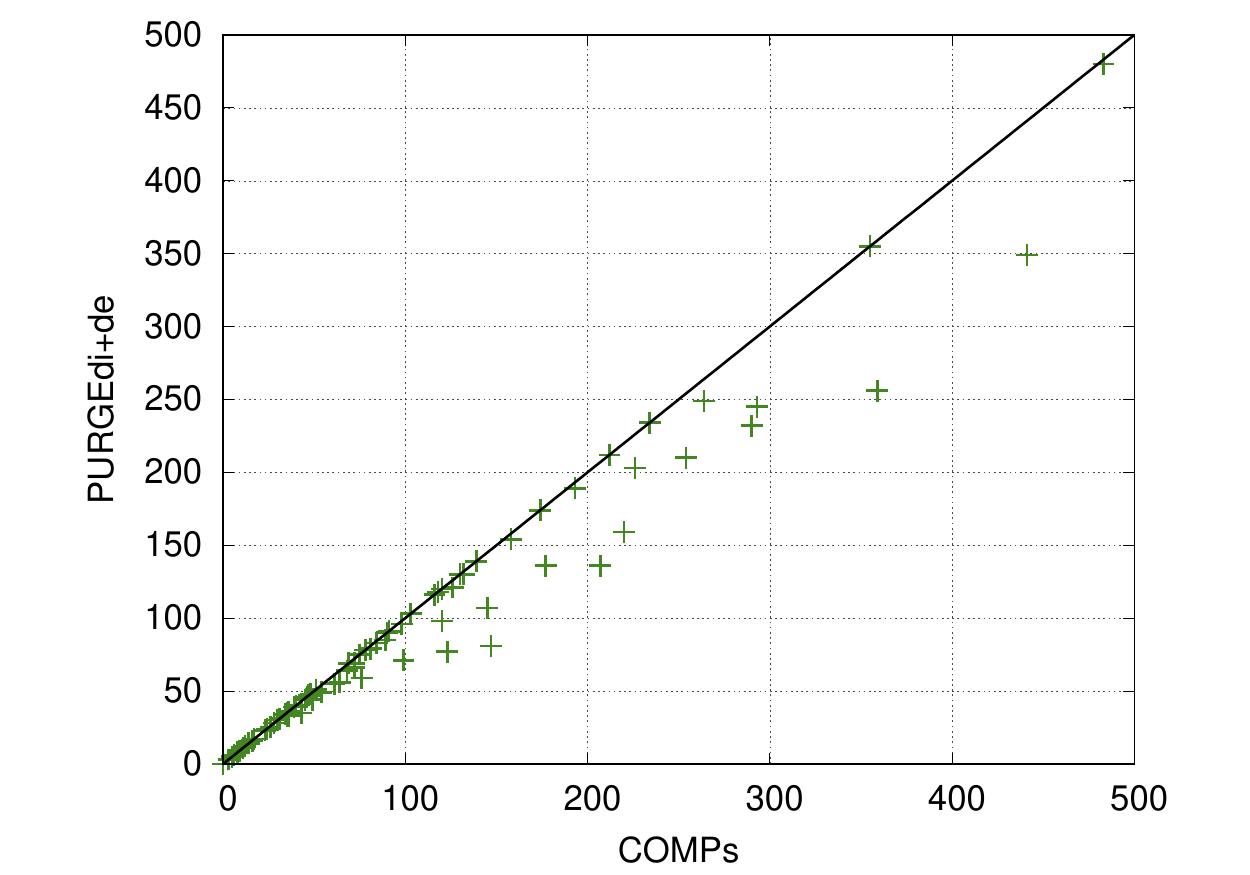}
      \caption{\PurgeDiDe vs.\ \CompSchewe}
      \label{fig:purge}
  \end{subfigure}
~
  \begin{subfigure}[b]{0.48\textwidth}
      \hspace*{-8mm}
      \includegraphics[width=70mm]{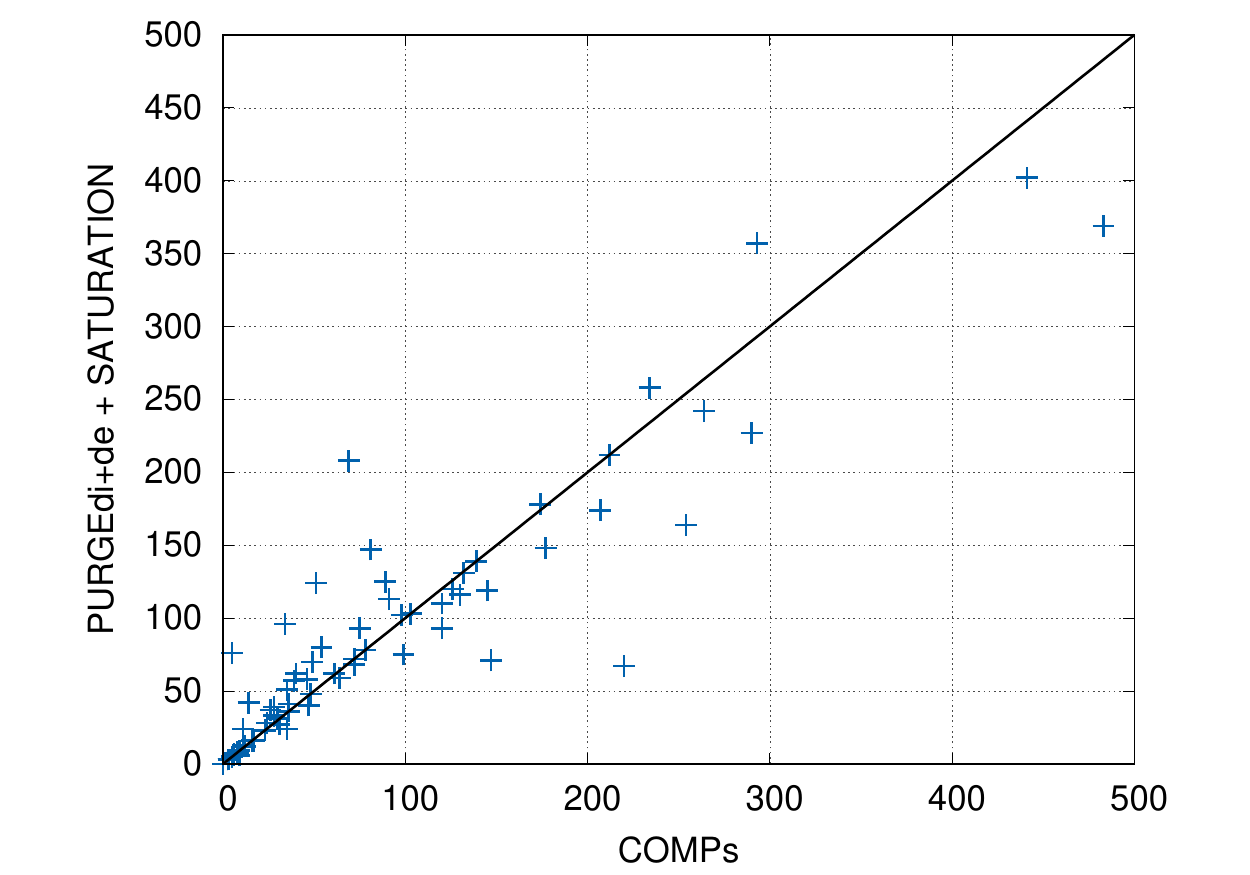}
      \caption{\PurgeDiDe\!\!\!+\Saturate vs.\ \CompSchewe}
      \label{fig:sat-del}
  \end{subfigure}
  \caption{Comparison of the number of states of complement BAs generated by
    \CompSchewe and our optimizations (lower is better)}
  \label{fig:res-del-sim}
\end{figure}


\vspace{-0.0mm}
\section{Related Work}\label{sec:related}
\vspace{-0.0mm}
BA complementation has a~long research track.
Known approaches can be roughly classified into
Ramsey-based~\cite{sistla1987complementation},
determinization-based~\cite{safra1988complexity,piterman2006nondeterministic},
rank-based~\cite{Schewe09},
slice-based~\cite{kahler2008complementation,vardi2008automata},
learning-based~\cite{li2018learning}, and the recently proposed subset-tuple
construction~\cite{allred2018simple}.
Those approaches build on top of different concepts of capturing words
accepted by a~complement automaton.
Some concepts can be translated into others, such as the slice-based
approach, which can be translated to the rank-based
approach~\cite{vardi2013unifying}.
Such a~translation can help us get a deeper
understanding of the BA complementation problem and the relationship between
optimization techniques for different complementation algorithms.

Because of the high computational complexity of complementing a~BA, and,
consequently, also checking BA inclusion and universality (which use
complementation as their component), there has been some effort to develop
heuristics that help to reduce the number of explored states in practical cases.
The most prominent ones are heuristics that leverage various notions of
simulation relations, which often provide a good compromise between the overhead
they impose and the achieved state space reduction.
Direct~\cite{bustan2003simulation,somenzi2000efficient},
delayed~\cite{EtessamiWS05},
fair~\cite{EtessamiWS05},
their variants for alternating B\"{u}chi automata~\cite{fritz-alternating-sim}, and
multi-pebble simulations~\cite{etessami2002hierarchy} are the best-studied
relations of this kind.
Some of the relations can be used quotienting, but also for pruning transitions
entering simulation-smaller states (which may cause some parts of the BA to
become inaccessible).
A~series of results in this direction was recently developed by Clemente and
Mayr~\cite{clemente2011buchi,MayrC13,mayr2019efficient}.

Not only can the relations be used for reducing the size of the input BA, they
can also be used for under-approximating inclusion of languages of states. For
instance, during a~BA inclusion test $\langof{\A_S} \stackrel{?}{\subseteq}
\langof{\A_B}$, if every initial state of~$\A_S$ is simulated by an initial
state of~$\A_B$, the inclusion holds and no complementation needs to be
performed. But simulations can also be used to reduce the explored state space
within, e.g., the inclusion check itself, for instance in the context of
Ramsey-based algorithms~\cite{abdulla2010simulation,abdulla2011advanced}.
Ramsey-based complementation algorithms~\cite{sistla1987complementation} in the
worst case produce $2^{\bigOof{n^2}}$ states, which is a~significant gap from
the lower bound of Michel~\cite{michel1988complementation} and
Yan~\cite{yan2006lower}.
The Ramsey-based construction was, however, later improved by Breuers
et al.~\cite{breuers-improved-ramsey} to
match the upper bound $2^{\bigO(n \log n)}$.
The way simulations are applied in the Ramsey-based approach is fundamentally
different from the current work, which is based on rank-based construction.
Taking universality checking as an example, the algorithm checks if the language
of the complement automaton is empty.
They run the complementation algorithm and the emptiness check together, on the
fly, and
during the construction check if a~macrostate with a~larger language has been
produced before; if yes, then they can stop the search from the language-smaller
macrostate.
Note that, in contrast to our approach, their algorithm does not produce the
complement automaton.

%
%

\vspace{-0.0mm}
\section{Conclusion and Future Work}\label{sec:conclusion}
\vspace{-0.0mm}

We developed two novel optimizations of the rank-based complementation algorithm for
B\"{u}chi automata that are based on leveraging direct and delayed simulation
relations to reduce the number of states of the complemented automaton.
The optimizations are directly usable in rank-based BA inclusion and
universality checking.
We conjecture that the decision problem of checking BA language inclusion might
also bring another opportunities for exploiting simulation, such as in a~similar
manner as in~\cite{AbdullaCHMV10}.
Another, orthogonal, directions of future work are (i)~applying simulation in
other than the rank-based approach (in addition to the particular use
within~\cite{abdulla2010simulation,abdulla2011advanced}), e.g., complementation
based on Safra's construction~\cite{safra1988complexity}, which, according to
our experience, often
produces smaller complements than the rank-based procedure, (ii) applying our
ideas within determinization constructions for BAs, and (iii)~generalizing
our techniques for richer simulations, such as the multi-pebble
simulation~\cite{etessami2002hierarchy} or various look-ahead
simulations~\cite{MayrC13,mayr2019efficient}.
Since the richer simulations are usually harder to compute, it would be
interesting to find the sweet spot between the overhead of simulation
computation and the achieved state space reduction.

\subsubsection*{Acknowledgement}

We thank the anonymous reviewers for their helpful comments on how to improve
the exposition in this paper.
This work was supported by
the Ministry of Science and Technology of Taiwan project 106-2221-E-001-009-MY3
the Czech Science Foundation project 19-24397S,
the FIT BUT internal project FIT-S-17-4014,
and The Ministry of Education, Youth and Sports from the
National Programme of Sustainability (NPU~II) project IT4Innovations
excellence in science---LQ1602.

\bibliographystyle{splncs04}
\bibliography{bibliography}
%

\end{document}